\documentclass[11pt]{llncs}
\usepackage{amsmath,amssymb}
\addtocounter{MaxMatrixCols}{1} % Default is 10, we need 11.
\usepackage{graphicx,booktabs,balance,enumitem}
\usepackage{pstricks,pst-node,pst-tree}

\begin{document}
\pagestyle{plain}
\title{On affine variety codes from the Klein quartic}

\author{Olav Geil\inst{1}%
\and Ferruh \"{O}zbudak\inst{1,2}}%
\authorrunning{Geil, \"{O}zbudak}
\tocauthor{Olav Geil, Ferruh \"{O}zbudak}
\institute{Department of Mathematical Sciences, 
    Aalborg University, Denmark\\
    \email{olav@math.aau.dk}\\
        \and
    Department of Mathematics and Institute of Applied Mathematics,
    Middle East Technical University, Turkey\\
    \email{ozbudak@metu.edu.tr}}

\maketitle

\begin{abstract}
We study a family of primary affine variety codes defined from the Klein
quartic. The duals of these codes have previously been  treated in~\cite[Ex.\
3.2]{KFR}. Among the codes that we construct almost all have parameters as
good as the best known codes according to~\cite{grassl} and in the
remaining few cases the parameters are almost as good. To establish the
code parameters we apply the footprint bound~\cite{onorin,geilhoeholdt} from Gr\"{o}bner basis
theory and for this purpose we develop a new method where we inspired by
Buchberger´s algorithm perform a series of symbolic computations. 
\end{abstract}

\section{Introduction}
Affine variety codes~\cite{lax} are codes defined by evaluating multivariate
polynomials at the points of an affine variety. Despite having a
simple description such codes constitute the entire class of linear
codes~\cite[Pro.\ 1]{lax}. Given a description of a code as an
affine variety code it is easy to determine the length $n$ and
dimension $k$, but no simple general method is known which
easily estimates the minimum distance $d$. Of course such methods exists
for particular classes of affine variety codes. For instance the Goppa bound for
one-point algebraic geometric codes extends to an improved bound on the more
general class of order domain codes~\cite{handbook,bookAG}, and in larger generality the
Feng-Rao bounds and their variants can be successfully applied to many
different types of codes\cite{FR24,FR1,FR2,KFR,salazar,bookAG,geil2013improvement}. In this paper we consider a particular
family of primary affine variety codes for which none of the above
mentioned bounds provide accurate information. More precisely we
consider primary codes defined from the Klein quartic using the same 
weighted degree lexicographic ordering as in~\cite[Ex.\ 3.2]{KFR} where
they studied the corresponding dual codes. A common property of the
Feng-Rao bound for primary codes and its variants are that they can be viewed\cite{bookAG,geil2013improvement} as
consequences of the footprint bound~\cite{onorin,geilhoeholdt} from Gr\"{o}bner basis theory. To establish more
accurate information for the codes under consideration it is therefore natural to try to apply the
footprint bound in a more direct way, which is exactly what we do
in the present paper using ingredients from Buchberger's
algorithm and by considering an exhaustive number of special cases. Our analysis reveals that the codes under
consideration are in most cases as good as the best known codes
according to~\cite{grassl} and for the remaining few cases the
minimum distance is only one less than the best known codes of the
same dimension.\\

The paper is organized as follows. In Section~\ref{secaffine} we
introduce the footprint of an ideal and define  
affine variety codes. We then describe how the footprint bound
can be applied to determine the Hamming weight of a code word. Then in
Section~\ref{secnew} we apply symbolic computations leading to
estimates on the minimum distance on each of the considered codes the
information of which we collect in Section~\ref{seccodes}. 
\section{Affine variety codes and the footprint bound}\label{secaffine}
The footprint (also called the delta-set) is defined as follows:
\begin{definition}
Given a field $k$, a monomial ordering $\prec$ and an ideal $J
\subseteq k[X_1, \ldots , X_m]$ the footprint of $J$ is
\begin{eqnarray}
\Delta_{\prec}(J)&=&\{ M \mid M {\mbox{ is a monomial which is not
                     leading monomial }}\nonumber \\
&& {\mbox{ \ \ \ \ \ \ \ \ \ \ \ \ \ \ \  \ \ \ \ \ \ \ \ \ \ \ \ \ \
   \  \ \ \ \ \ \ \ \ \ \ of any polynomial in }} J\}\nonumber
\end{eqnarray}
\end{definition}
From~\cite[Prop.\ 7, Sec.\ 5.3]{clo4} we have the following well-known result.
\begin{theorem}\label{thebasis}
Let the notation be as in the above definition. The set $$\{M+J \mid M
\in \Delta_\prec(J)\}$$ is a basis for $k[X_1, \ldots , X_m]/J$ as a
vector space over $k$.
\end{theorem}
Recall that by definition a Gr\"{o}bner basis is a finite basis for the ideal
$J$ from which one can easily determine the footprint. Concretely a
monomial is a leading monomial of some polynomial in the ideal if and
only if it is divisible by a leading monomial of some polynomial in
the Gr\"{o}bner basis. The following corollary is an instance of the
more general footprint bound~\cite{onorin}.
\begin{corollary}\label{corkor}
Let $I \subseteq {\mathbb{F}}_q[X_1, \ldots , X_m]$ be an ideal and $I_q=I+\langle
X_1^q-X_1, \ldots , X_m^q-X_m  \rangle$. The variety of $I_q$ is of
size $\# \Delta_\prec(I_q)$ for any monomial ordering $\prec$.
\end{corollary}
\begin{proof}
Let the variety of $I_q$ be $\{P_1, \ldots , P_n\}$ with $P_i \neq
P_j$ for $i \neq j$. The field ${\mathbb{F}}_q$ being perfect, the ideal $I_q$ is radical because it contains a
univariate square-free polynomial in each variable and by the
ideal-variety correspondence therefore $I_q$ is in fact the vanishing
ideal of $\{P_1, \ldots , P_n\}$. Therefore the evaluation map
${\mbox{ev}}: {\mathbb{F}}_q[X_1, \ldots , X_m]/I_q \rightarrow
{\mathbb{F}}_q^n$ given by ${\mbox{ev}}(F+I_q)=(F(P_1), \ldots ,
F(P_n))$ is injective. On the other hand the evaluation map is also
surjective which is seen by applying Lagrange interpolation. We have
demonstrated that ${\mbox{ev}}$ is a bijection and the corollary
follows from Theorem~\ref{thebasis}. \qed
\end{proof}
We are
now ready to define primary affine variety codes formally.

\begin{definition} Let the notation be as in the proof of
  Corollary~\ref{corkor}. Given an ideal $I \subseteq {\mathbb{F}}_q[X_1, \ldots , X_m]$ and a
monomial ordering $\prec$ choose $L \subseteq
\Delta_{\prec}(I_q)$. Then $$C(I,L)={\mbox{Span}}_{\mathbb{F}_q}\{{\mbox{ev}}(M+I_q)
\mid M \in L\}$$
is called a primary affine variety code.
\end{definition}
From the above discussion it is clear that $C(I,L)$ is a code of
length $n=\# \Delta_\prec(I_q)$ and dimension $k=\# L$. Given a code word $\vec{c}={\mbox{ev}}(F+I_q)$ then by
Corollary~\ref{corkor} we have 
\begin{eqnarray}
w_H(\vec{c})&=&n- \# \Delta_{\prec_w}(\langle F \rangle +I_q)= \#
                \Delta_{\prec_w}(I_q) \cap {\mbox{lm}}(\langle
                F\rangle +I_q)= \# \Box_{\prec_w}(F),\nonumber
\end{eqnarray}
where $\Box_{\prec_w}(F):=\Delta_{\prec_w}(I_q) \cap
{\mbox{lm}}(\langle F\rangle +I_q)$. 
Reducing a polynomial modulo a Gr\"{o}bner basis for $I_q$ one
obtains a (unique) polynomial which has support in the footprint
$\Delta(I_q)$ (this is the result behind Theorem~\ref{thebasis}). Hence we shall always assume that $F$ is of this
form. In the rest of the paper we concentrate on estimating $\#
\Box_{\prec}(F)$ using only information on the leading monomial. We do
this for a concrete class of codes defined from the Klein quartic, but
the method that we describe can be applied to any affine variety code
of moderate dimension. In particular it can be applied whenever the
length of the codes are moderate.
\section{Code words from the Klein curve}\label{secnew}
In the remaining part of the paper $I$ will always be the ideal 
$$I=\langle Y^3+X^3Y+X\rangle \subseteq {\mathbb{F}}_8[X,Y]$$
and consequently $I_8=\langle Y^3+X^3Y+X,X^8+X,Y^8+Y\rangle$. The
corresponding variety\footnote{As we treat 
codes at a theoretical level we shall not need detailed information on the variety, but we find it interesting to note
that besides one point being $(0,0)$ the remaining points correspond
to the Fano plane by identifying each non-zero element in
${\mathbb{F}}_8$ with a vertex. Every non-zero $a$ now defines a line
consisting of all $b$s such that $(a,b)$ is in the variety.} is of size $22$, hence we write it as $\{P_1,
\ldots , P_{22}\}$. The evaluation map then becomes
${\mbox{ev}}(F+I_8)=(F(P_1), \ldots , F(P_{22}))$.\\

As monomial ordering we choose the same ordering as in~\cite[Ex.\
3.2]{KFR}, namely the weighted degree
lexicographic ordering $\prec_w$ given by the rule that  $X^\alpha
Y^\beta \prec_w X^\gamma Y^\delta$ if either ({\it{i}}) or ({\it{ii}}) below holds
$$ {\mbox{ ({\it{i}}) }} 2 \alpha + 3 \beta < 2 \gamma + 3 \delta, {\mbox{ \ \
    ({\it{ii}}) }}2 \alpha + 3 \beta = 2 \gamma + 3 \delta {\mbox{ but }} \beta < \delta.$$
By inspection $\{Y^3+X^3Y+X,X^8-X,X^7Y+Y\}$ is a Gr\"{o}bner basis for
$I_8$ with respect to $\prec_w$. Hence, the footprint $\Delta_{\prec_w}(I_8)$ and the
corresponding weights are as in
Figure~\ref{figfirstone}. We remind the reader that for $L \subseteq
\Delta_{\prec_w}(I_8)$ the code 
$C(I,L)$ equals ${\mbox{ev}}({\mbox{Span}}_{{\mathbb{F}}_8}(L)+I_8)$
which is of length $n=22$ and 
dimension $k=\#L$. 
\begin{figure}
\begin{center}
$$
\begin{array}{cccccccc}
Y^2&XY^2&X^2Y^2&X^3Y^2&X^4Y^2&X^5Y^2&X^6Y^2\\
Y&XY&X^2Y&X^3Y&X^4Y&X^5Y&X^6Y\\
1&X&X^2&X^3&X^4&X^5&X^6&X^7\\
\ \\
6&8&10&12&14&16&18\\
3&5&7&9&11&13&15\\
0&2&4&6&8&10&12&14
\end{array}
$$
\end{center}
\caption{The footprint $\Delta_{\prec_w}(I_8)$ with corresponding
  weights.}
\label{figfirstone}
\end{figure}

Our
method to estimate $\# \Box_{\prec_w}(F)$ (which corresponds to
estimating the Hamming weight of the corresponding code word)
consists in two parts. First we
observe that all monomials in $\Delta_{\prec_w}(I_8)$ divisible by the
leading monomial of $F$ are in $\Box_{\prec_w}(F)$. In the second part we then for a number of exhaustive
special cases find more monomials in
$\Box_{\prec_w}(F)$ by establishing clever combinations of polynomials that
we already know are in $\langle F \rangle +I_q$. To describe how such combinations are derived
we will need the following notation. Consider polynomials $S(X,Y)$, $D(X,Y)$ and $R(X,Y)$. By 
\begin{eqnarray}
S(X,Y) \overset{D(X,Y)}{\longrightarrow}
 R(X,Y)\label{eqreduce}
\end{eqnarray}
we shall indicate that $R(X,Y)=S(X,Y)-Q(X,Y) D(X,Y)$ for some
polynomial $Q(X,Y)$. The important fact -- which we shall use
frequently throughout the paper -- is that $R(X,Y) \in \langle
S(X,Y),D(X,Y)\rangle$. Observe that although we will always use the
above ``operation'' to decrease the leading monomial (meaning that
${\mbox{lm}}(R)\prec {\mbox{lm}}(S)$ ), we may still have monomials
left in the support of $R(X,Y)$ which are divisible by the leading
monomial of $D(X,Y)$. Hence, (\ref{eqreduce}) does not necessarily
correspond to the usual (full) division as described
in~\cite[Sec. 2.3]{clo4}.

\begin{remark}\label{remrom}
The Feng-Rao bound can be applied to any affine variety code; but it
works most efficiently when the ideal $I$ and the monomial ordering $\prec$ under
consideration satisfy
the order domain conditions \cite[Sec.\ 7]{bookAG}. That is, 
\begin{enumerate}
\item The ordering $\prec$
must be a weighted degree lexicographic ordering (or in larger
generality a generalized weighted 
degree ordering~\cite[Def.\ 8]{bookAG}).
\item  A Gr\"{o}bner
basis for $I$ must exist with the property that any polynomial in it
contains in its support (exactly)
two monomials of the highest weight. 
\item No two different
monomials in $\Delta_\prec(I)$ are of the same weight. 
\end{enumerate}
In such cases the method often establishes many more monomials in
$\Box_{\prec}(F)$ than those divisible by the leading monomial of
$F$. In~\cite{geil2013improvement} an improved Feng-Rao bound was
presented which treats in addition efficiently certain
families of cases where the conditions 1.\ and 2.\ are satisfied, but
3. is not. Even though the ideal and monomial ordering studied in the present
section exactly satisfy conditions 1.\ and 2., but not 3, the improved Feng-Rao
bound produces the same information as the Feng-Rao
bound in this case. By inspection both methods only ``detect'' monomials divisible by the leading monomial of
$F$ as being members of $\Box_{\prec_w}(F)$.
\end{remark}

Below we treat the $22$ different possible
leading monomials -- corresponding to the different members of
$\Delta_{\prec_w}(I_8)$ -- one by one. For simplicity, we shall in our
calculations always assume that the leading coefficient of $F$ is $1$
which is not really a restriction as our goal is to estimate Hamming weights.

\subsection{Leading monomial equal to  $Y$}

Consider $\vec{c}={\mbox{ev}}(F+I_8)$ where
$F(X,Y)=Y+a_1X+a_2$. Clearly 
$$\{ Y, Y^2, XY, XY^2, \ldots ,X^6Y,X^6Y^2\} \subset
\Box_{\prec_w}(F).$$
We next establish more monomials in $\Box_{\prec_w}(F)$ under
different conditions on the coefficients $a_1, a_2$. Consider 
\begin{eqnarray}
&&Y^2F(X,Y) \nonumber\\
&\overset{Y^3+X^3Y+X}{\longrightarrow} &X^3Y+a_1XY^2+a_2Y^2+X\nonumber
  \\
&\overset{F(X,Y)}{\longrightarrow}&a_1X^4+(a_1^3+a_2)X^3+a_1^2a_2X^2+(a_1
                                    a_2^2+1)X+a_2^3.\nonumber
\end{eqnarray}
If $a_1 \neq 0$ then we have $$\{X^4,X^5,X^6,X^7\} \subset \Box_{\prec_w}(F).$$ Next assume $a_1=0$. If $a_2\neq 0$ then we obtain
$$\{X^3,X^4,X^5,X^6,X^7\} \subset \Box_{\prec_w}(F).$$
Finally, assume $a_1=a_2=0$ in which case we have
$$\{X,X^2,X^3X^4,X^5,X^6,X^7\} \subset
\Box_{\prec_w}(F).$$
In conclusion we have shown that $\Box_{\prec_w}(F)$ contains at least
$14+4=18$ elements which implies
$w_H(\vec{c})\geq 18$.

\subsection{Leading monomial equal to $Y^2$}

Consider a codeword $\vec{c}={\mbox{ev}}(F+I_8)$ where 
$$F(X,Y)=Y^2+a_1X^3+a_2XY+a_3X^2+a_4Y+a_5X+a_6.$$
Independently of the coefficients $a_1, \ldots , a_6$ we see that 
\begin{eqnarray}
\{Y^2, XY^2, \ldots , X^6Y^2\} \subset \Box_{\prec_w}(F).\label{eqtrekant1} 
\end{eqnarray}
We next consider an exhaustive series of conditions under which we establish more
monomials in $\Box_{\prec_w}(F)$. We have
\begin{eqnarray}
&&YF(X,Y)\nonumber\\
&\overset{Y^3+X^3Y+X}{\longrightarrow}&(a_1+1)X^3Y+a_2XY^2+a_3X^2Y
                                        \nonumber \\
&&+a_4Y^2+a_5XY+a_6Y+X. \label{eqenting1}
\end{eqnarray}
If $a_1 \neq 1$ then the leading monomial of the last polynomial
becomes $X^3Y$ and consequently
\begin{eqnarray}
\{X^3Y,X^4Y,X^5Y,X^6Y \}\in \Box_{\prec_w}(F). \label{eqtrekant2}
\end{eqnarray}
Continuing the calculations for this case we obtain:
\begin{eqnarray}
&&Y((a_1+1)X^3Y+a_2XY^2+a_3X^2Y+a_4Y^2+a_5XY+a_6Y+X)\nonumber \\
&\overset{F(X,Y)}{\longrightarrow}&(a_1+1)(a_1X^6+a_2X^4Y+a_3X^5+a_4X^3Y+a_5X^4+a_6X^3)\nonumber
  \\
&&+a_2XY^3+a_3X^2Y^2+a_4Y^3+a_5XY^2+a_6Y^2+XY.\nonumber
\end{eqnarray}
If $a_1 \neq 0$ then we also have 
$$\{X^6,X^7\} \subset \Box_{\prec_w}(F).$$
Assuming next that $a_1=0$ the above expression
becomes
\begin{eqnarray}
&&a_2X^4Y+a_3X^5+a_4X^3Y+a_5X^4+a_6X^3+a_2XY^3\nonumber \\
&&+a_3X^2Y^2+a_4Y^3+a_5XY^2+a_6Y^2+XY\nonumber \\
&\overset{Y^3+X^3Y+X}{\longrightarrow}&a_3X^5+a_4X^3Y+a_5X^4+a_6X^3+a_3X^2Y^2+a_4Y^3\nonumber
  \\
&&+a_5XY^2+a_6Y^2+XY+a_2X^2\nonumber \\
&\overset{F(X,Y)}{\longrightarrow}&a_3X^5+a_5X^4+a_6X^3+a_3a_2X^3Y+a_3^2X^4+a_3a_4X^2Y
\nonumber
  \\
&&+a_3a_5X^3+a_3a_6X^2+a_5XY^2+a_6Y^2+XY+a_2X^2.\nonumber
\end{eqnarray}
If $a_3\neq 0$ then
$$\{X^5,X^6,X^7\} \subset \Box_{\prec_w}(F).$$
Hence, continuing under the assumption $a_3=0$ we are left with
\begin{eqnarray}
&&a_5X^4+a_6X^3+a_5XY^2+a_6Y^2+XY+a_2X^2 \nonumber \\
&\overset{F(X,Y)}{\longrightarrow}&a_5X^4+a_6X^3+a_5a_2X^2Y+a_5a_4XY+a_5^2X^2+a_5a_6X+a_6Y^2\nonumber
                                    \\
&&+XY+a_2X^2.\nonumber
\end{eqnarray}
if $a_5 \neq 0$ then 
$$\{X^4,X^5,X^6,X^7\} \subset \Box_{\prec_w}(F).$$
Hence, assume $a_5 =0$ and we are left with
\begin{eqnarray}
&&a_6X^3+a_6Y^2+XY+a_2X^2\nonumber \\
&\overset{F(X,Y)}{\longrightarrow}&a_6X^3+a_6a_2XY+a_6a_4Y+a_6^2+XY+a_2X^2.\nonumber 
\end{eqnarray}
If $a_6\neq 0$ then 
$$\{X^3,X^4,X^5, X^6, X^7\}\subset \Box_{\prec_w}(F).$$
If on the other hand $a_6=0$ then we are left with $XY+a_2X^2$ in
which case we obtain
$$\{XY,X^2Y\} \subset \Box_{\prec_w}(F).$$
In conclusion, for the case $a_1 \neq 1$ we obtained in addition to
the elements in~(\ref{eqtrekant1}) the elements in~(\ref{eqtrekant2})
and at least $2$ more. That is, in addition to the elements
in~(\ref{eqtrekant1}) at least $6$ more.\\
Assume in the following that $a_1=1$ and continue the reduction from~(\ref{eqenting1}) 
\begin{eqnarray}
&\overset{F(X,Y)}{\longrightarrow}&a_2X^4+(a_3+a_2^2)X^2Y+(a_2a_3+a_4)X^3+a_5XY
                                    \nonumber
    \\
&&+(a_2a_5+a_3a_4)X^2+
(a_6+a_4^2)Y+(1+a_4a_5)X+a_4a_6. \label{eqcontinuingfrom1}
\end{eqnarray}
If $a_2 \neq 0$ then 
\begin{eqnarray}
\{X^4,X^5,X^7,X^4Y,X^5Y,X^6Y\} \subset \Box_{\prec_w}(F). \nonumber
\end{eqnarray}
Next assume $a_2=0$. If $a_3\neq 0$ then 
\begin{eqnarray}
\{X^2Y,X^3Y,X^4Y,X^5Y, X^6Y\} \subset
  \Box_{\prec_w}(F). \label{eqtrekant3}
\end{eqnarray}
Continuing the reduction under the assumption $a_3\neq 0$ we
multiply~(\ref{eqcontinuingfrom1}) by $Y$ and continue the reduction:
\begin{eqnarray}
&&a_3X^2Y^2+a_4X^3Y+a_5XY^2+a_3a_4X^2Y+(a_6+a_4^2)Y\nonumber \\
&&+(1+a_4a_5)X+a_4a_6\nonumber
  \\
&\overset{F(X,Y)}{\longrightarrow}&a_3X^5+a_3^2X^4+a_3a_4X^2Y+a_3a_5X^3+a_3a_6X^2+a_4X^3Y+a_5XY^2\nonumber
  \\
&&+a_3a_4X^2Y+(a_6+a_4^2)Y+(1+a_4a_5)X+a_4a_6.\nonumber
\end{eqnarray}
As $a_3 \neq 0$ we obtain in addition to 
(\ref{eqtrekant1}) and (\ref{eqtrekant3}) that 
$$\{ X^5,X^6,X^7\} \subset \Box_{\prec_w}(F).$$ 
That is, in addition to (\ref{eqtrekant1}) we found in total  $8$
  more elements in $\Box_{\prec_w}(F)$.

Next assume $a_3=0$ and continue from~(\ref{eqcontinuingfrom1}). If $a_4 \neq 0$ then
\begin{eqnarray}
\{X^3,X^4,X^5,X^6,X^7,X^3Y,X^4Y, X^5Y, X^6Y\} \subset \Box_{\prec_w}(F).\nonumber 
\end{eqnarray}
Next assume $a_4=0$. if $a_5 \neq 0$ then 
\begin{eqnarray}
\{XY,X^2Y, X^3Y,X^4Y,X^5Y, X^6Y\} \subset \Box_{\prec_w}(F). \nonumber
\end{eqnarray}
Hence, assume $a_5=0$. If $a_6 \neq 0$ then 
\begin{eqnarray}
\{Y,XY,X^2Y,X^3Y,X^4Y,X^5Y, X^6Y\} \subset \Box_{\prec_w}(F). \nonumber
\end{eqnarray}
Finally, assume $a_6=0$. But then
\begin{eqnarray}
\{X,X^2,X^3,X^4,X^5,X^6,X^7,
XY,X^2Y,X^3Y,X^4Y,X^5Y, X^6Y\} \subset \Box_{\prec_w}(F).\nonumber
\end{eqnarray}
In conclusion, we have at least $7+\min\{ 6,6,8,9,6,7,13\}=13$
monomials in $\Box_{\prec_w}(F)$ and therefore $w_H(\vec{c}) \geq 13$. 

\subsection{Leading monomial equal to $XY$}

Consider $\vec{c}={\mbox{ev}}(F+I_8)$ where 
$$F(X,Y)=XY+a_1X^2+a_2Y+a_3X+a_4.$$
For sure
\begin{eqnarray}
\{ XY,X^2Y,X^3Y,X^4Y,X^5Y,X^6Y,{\mbox{ \ \ \ \ \ \ \ \ \ \ \ }}\nonumber \\
XY^2,X^2Y^2,X^3Y^2,X^4Y^2,X^5Y^2,X^6Y^2\}& \subset \Box_{\prec_w}(F). \label{eqsnipsnapsnude}
\end{eqnarray}
We next consider an exhaustive series of conditions under which we
establish more monomials in $\Box_{prec_w}(F)$. We have
\begin{eqnarray}
&&Y^2F(X,Y) \nonumber \\
&\overset{Y^3+X^3Y+X}{\longrightarrow}&a_1X^5+a_3X^4+a_4X^3\nonumber \\
&&+a_1(a_1^2X^4+a_2^2Y^2+a_3^2X^2+a_4^2) \nonumber \\
&&+a_3XY^2+a_4Y^2+X^2+a_2X.\nonumber
\end{eqnarray}
If $a_1 \neq 0$ then 
$$\{X^5,X^6,X^7\} \subset \Box_{\prec_w}(F).$$
Hence, assume $a_1=0$ and continue the reduction:
\begin{eqnarray}
&\overset{F(X,Y)}{\longrightarrow}&a_3a_2Y^2+a_3^2XY+a_3a_4Y+a_3X^4+a_4X^3+a_4Y^2+X^2+a_2X. \nonumber
\end{eqnarray}
If $a_3 \neq 0$ then 
$$\{X^4,X^5,X^6,X^7\} \subset \Box_{\prec_w}(F).$$
Hence, assume $a_3=0$ in which case the above becomes
$$a_4Y^2+a_4X^3+X^2+a_2X.$$
If $a_4=0$ then
$$\{X^2,X^3,X^4,X^5,X^6,X^7\} \subset \Box_{\prec_w}(F).$$
Hence, assume $a_4\neq 0$, in which case we have
$$\{Y^2\} \subset \Box_{\prec_w}(F).$$
We continue the calculations to add more elements. We have:
\begin{eqnarray}
X^2(a_4Y^2+a_4X^3+X^2+a_2X) \overset{F(X,Y)}{\longrightarrow}a_4X^5+a_2^2Y^2+a_2X+a_4^2.\nonumber
\end{eqnarray}
But then 
$$\{X^5,X^6,X^7\} \subset \Box_{\prec_w}(F).$$
That is, for the case $a_4 \neq 0$ $\Box_{\prec_w}(F)$ contains in addition
to~(\ref{eqsnipsnapsnude}) at least $1+3=4$ more monomials.\\

In conclusion $w_H(\vec{c})\geq 12+\min \{3,4,6,4\}=15$,
and if $a_1=0$ then $w_H(\vec{c}) \geq 16$.

\subsection{Leading monomial equal to $X^2Y$}
Consider $\vec{c}={\mbox{ev}}(F+I_8)$ where 
\begin{eqnarray}
F(X,Y)&=&X^2Y+a_1Y^2+a_2X^3+a_3XY+a_4X^2+a_5Y+a_6X+a_7.\nonumber
\end{eqnarray}
For sure
\begin{eqnarray}
\{
  X^2Y,X^3Y,X^4Y,X^5Y,X^6Y,X^2Y^2,X^3Y^2,X^4Y^2,X^5Y^2,X^6Y^2\}\nonumber
  \\
\subset \Box_{\prec_w}(F).\label{eqsnabeleins6} 
\end{eqnarray}
We next consider an exhaustive series of conditions under which we establish more
monomials in $\Box_{\prec_w}(F)$. We have
\begin{eqnarray}
&&Y^2F(X,Y) \nonumber \\
&=&X^2Y^3+a_1Y^4+a_2X^3Y^2+a_3XY^3+a_4X^2Y^2+a_5Y^3\nonumber \\
&&+a_6XY^2+a_7Y^2\nonumber
  \\
&\overset{Y^3+X^3Y+X}{\longrightarrow}&X^5Y+a_1X^3Y^2+a_2X^3Y^2+a_3X^4Y+a_4X^2Y^2+a_5X^3Y\nonumber
                                        \\
&&+a_6XY^2+a_7Y^2+X^3+a_1XY+a_3X^2+a_5X\nonumber \\
&\overset{F(X,Y)}{\longrightarrow}&a_2X^6+a_4X^5+a_6X^4+a_7X^3+a_2X^3Y^2+a_4X^2Y^2\nonumber
                                    \\
&&+a_6XY^2+a_7Y^2+X^3+a_1XY+a_3X^2+a_5X.\nonumber
\end{eqnarray}
If $a_2\neq 0$ then 
$$\{X^6, X^7\} \subset \Box_{\prec_w}(F).$$
Hence, assume $a_2=0$, in which case we have
\begin{eqnarray}
&&a_4X^5+a_6X^4+a_7X^3+a_4X^2Y^2+a_6XY^2+a_7Y^2+X^3\nonumber \\
&&+a_1XY+a_3X^2+a_5X\nonumber
  \\
&\overset{F(X,Y)}{\longrightarrow}&a_4X^5+a_6X^4+a_7X^3+a_4Y(a_1Y^2+a_3XY+a_4X^2\nonumber
  \\
&&+a_5Y+a_6X+a_7)+a_6XY^2+a_7Y^2+X^3+a_1XY+a_3X^2+a_5X.\nonumber
\end{eqnarray} 
If $a_4 \neq 0$ then
$$\{X^5,X^6,X^7\} \subset \Box_{\prec_w}(F).$$
Hence, assume $a_4=0$ and continuer
\begin{eqnarray}
&&X(a_6X^4+a_7X^3+a_6XY^2+a_7Y^2+X^3+a_1XY+a_3X^2+a_5X)\nonumber \\
&\overset{F(X,Y)}{\longrightarrow}&a_6X^5+a_7X^4+a_6Y(a_1Y^2+a_3XY+a_5Y+a_6X+a_7)\nonumber
  \\
&&+a_7XY^2+X^4+a_1X^2Y+a_3X^3+a_5X^2.\nonumber
\end{eqnarray}
If $a_6 \neq 0$ then
$$\{X^5,X^6,X^7\} \subset \Box_{\prec_w}(F).$$
Hence, assume $a_6=0$, in which case we have
\begin{eqnarray}
&&X(a_7X^4+a_7XY^2+X^4+a_1X^2Y+a_3X^3+a_5X^2)\nonumber \\
&\overset{F(X,Y)}{\longrightarrow}&(a_7+1)X^5+a_7Y(a_1Y^2+a_3XY+a_5Y+a_7)\nonumber
  \\
&&+a_1X^3Y+a_3X^4+a_5X^3.\nonumber 
\end{eqnarray}
If $a_7 \neq 1$ then
$$\{X^5,X^6,X^7\} \subset \Box_{\prec_w}(F).$$
Hence, assume $a_7=1$ and continue the reduction
\begin{eqnarray}
&&a_1Y^3+a_3XY^2+a_5Y^2+Y+a_1X^3Y+a_3X^4+a_5X^3\nonumber \\
&\overset{Y^3+X^3Y+X}{\longrightarrow}&a_3XY^2+a_5Y^2+Y+a_3X^4+a_5X^3+a_1X\nonumber
\end{eqnarray}
which we multiply by $X$ before continuing reduction
\begin{eqnarray}
&&a_3X^2Y^2+a_5XY^2+XY+a_3X^5+a_5X^4+a_1X^2\nonumber \\
&\overset{F(X,Y)}{\longrightarrow}&a_3Y(a_1Y^2+a_3XY+a_5Y+1)\nonumber
  \\
&&+a_5XY^2+XY+a_3X^5+a_5X^4+a_1X^2.\nonumber
\end{eqnarray}
If $a_3 \neq 0$ then 
$$\{X^5,X^6,X^7\} \subset \Box_{\prec_w}(F).$$
Hence, assume $a_3=0$ and continue
\begin{eqnarray}
&&X(a_5XY^2+XY+a_5X^4+a_1X^2)\nonumber \\
&\overset{F(X,Y)}{\longrightarrow}&a_5(a_1Y^3+a_5Y^2+Y)+a_1Y^2+a_5Y+1+a_5X^5+a_1X^3.\nonumber
\end{eqnarray}
If $a_5\neq 0$ then 
$$\{X^5,X^6,X^7\} \subset \Box_{\prec_w}(F).$$
Hence, assume $a_5=0$ and multiply the resulting expression by $Y$
\begin{eqnarray}
&&Y(a_1Y^2+a_1X^3+1)\nonumber \\
&\overset{Y^3+X^3Y+X}{\longrightarrow}&Y+a_1X\nonumber
\end{eqnarray}
and we conclude
$$\{Y,Y^2,XY,XY^2\} \subset \Box_{\prec_w}(F).$$
In conclusion $w_H(\vec{c})\geq 10+\min \{2,3,3,3,3,3,4\}=12$, 
and if $a_2 =0$ then $w_H(\vec{c})\geq 13$.

\subsection{Leading monomial equal to $XY^2$}
Consider $\vec{c}={\mbox{ev}}(F+I_8)$ where 
\begin{eqnarray}
F(X,Y)&=&XY^2+a_1X^4+a_2X^2Y+a_3Y^2+a_4X^3\nonumber \\
&&+a_5XY+a_6X^2+a_7Y+a_8X+a_9.\nonumber
\end{eqnarray}
For sure
\begin{eqnarray}
&&\{ XY^2,X^2Y^2,X^3Y^2,X^4Y^2,X^5Y^2,X^6Y^2\} \subset \Box_{\prec_w}(F).\label{eqsnabeleins} 
\end{eqnarray}
We next consider an exhaustive series of conditions under which we establish more
monomials in $\Box_{\prec_w}(F)$. We have
\begin{eqnarray}
&&YF(X,Y) \nonumber \\
&\overset{Y^3+X^3Y+X}{\longrightarrow}
  &(a_1+1)X^4Y+a_2X^2Y^2+a_3Y^3+a_4X^3Y+a_5XY^2\nonumber \\
&&+a_6X^2Y+a_7Y^2+a_8XY+X^2+a_9Y.\label{eqtriangle1}
\end{eqnarray}
If $a_1 \neq 1$ then 
\begin{equation}
\{X^4Y,X^5Y,X^6Y\} \subset \Box_{\prec_w}(F).\label{eqsnabelzwei}
\end{equation}
Continuing the calculations for this case we obtain
\begin{eqnarray}
&&Y((a_1+1)X^4Y+a_2X^2Y^2+a_3Y^3+a_4X^3Y+a_5XY^2+a_6X^2Y\nonumber \\
&&+a_7Y^2+a_8XY+X^2+a_9Y)
   \nonumber \\
&\overset{F(X,Y)}{\longrightarrow}&(a+1)(a_1X^7+a_2X^5Y+a_3X^3Y^2+a_4X^6+a_5X^4Y+a_6X^5\nonumber
  \\
&&+a_7X^3Y+a_8X^4+a_9X^3)+a_2X^2Y^3+a_3Y^4+a_4X^3Y^2+a_5XY^3\nonumber
  \\
&&+a_6X^2Y^2+a_7Y^3+a_8XY^2+X^2Y+a_9Y^2.\nonumber \end{eqnarray}
If $a_1\neq 0$ then we also have 
$$\{X^7\} \subset \Box_{\prec_w}(F).$$
Assuming next that $a_1=0$ the above expression becomes
\begin{eqnarray}
&&a_2X^5Y+a_3X^3Y^2+a_4X^6+a_5X^4Y+a_6X^5+a_7X^3Y+a_8X^4\nonumber \\
&&+a_9X^3+a_2X^2Y^3+a_3Y^4+a_4X^3Y^2\nonumber
  \\
&&+a_5XY^3+a_6X^2Y^2+a_7Y^3+a_8XY^2+X^2Y+a_9Y^2\nonumber \\
&\overset{Y^3+X^3Y+X}{\longrightarrow}&a_4X^6+a_6X^5+a_8X^4+a_9X^3+a_3XY+a_4X^3Y^2\nonumber
  \\
&&+a_6X^2Y^2+a_8XY^2+X^2Y+a_9Y^2+a_2X^3+a_5X^2+a_7X\nonumber \\
&\overset{F(X,Y)}{\longrightarrow}&a_4X^6+a_6X^5+a_8X^4+a_9X^3+a_3XY\nonumber \\
&&+(a_4X^2+a_6X+a_8)(a_2X^2Y+a_3Y^2+a_4X^3+a_5XY+a_6X^2\nonumber \\
&&+a_7Y+a_8X+a_9)+X^2Y+a_9Y^2+a_2X^3+a_5X^2+a_7X.\nonumber
\end{eqnarray}
If $a_4\neq 0$ then
$$\{X^6,X^7\} \subset \Box_{\prec_w}(F).$$
Hence, we assume $a_4=0$. From the above expression we see that if
next $a_6\neq 0$ then 
$$\{X^5,X^6,X^7\}\subset \Box_{\prec_w}(F).$$
Hence, assume $a_6=0$. Investigating again the above expression we now
see that for $a_8 \neq 0$ it holds that 
$$\{X^4,X^5,X^6,X^7\} \subset \Box_{\prec_w}.$$
Continuing from the same expression, but now under the assumption that
$a_8=0$ we see that
$$\{X^2Y,X^3Y\} \subset \Box_{\prec_w}(F).$$
In conclusion, for the case $a_1 \neq 1$ we have in addition to
(\ref{eqsnabeleins}) and (\ref{eqsnabelzwei}) established at least one
more element in $\Box_{\prec_w}(F)$. That is, in addition to
(\ref{eqsnabeleins}) we have at least $4$ elements in
$\Box_{\prec_w}(F)$. Furthermore, if $a_1\neq 1$ and $a_1\neq 0$ then
we have at least one more element in addition in this set.\\

In the following we assume $a_1=1$ and continue the calculations from~(\ref{eqtriangle1}) as follows
\begin{eqnarray}
&\overset{Y^3+X^3Y+X}{\longrightarrow}&a_2X^2Y^2+(a_3+a_4)X^3Y+a_5XY^2+a_6X^2Y\nonumber
  \\
&&+a_7Y^2+a_8XY+X^2+a_9Y+a_3X. \nonumber\\
&\overset{F(X,Y)}{\longrightarrow}&a_2X^5+(a_2^2+a_3+a_4)X^3Y+(a_2a_3+a_5)XY^2+a_2a_4X^4
                                    \nonumber \\
&&+(a_2a_5+a_6)X^2Y+a_7Y^2+a_2a_6X^3+(a_2a_7+a_8)XY\nonumber \\
&&+(a_2a_8+1)X^2+a_9Y+(a_2a_9+a_3)X.\nonumber
\end{eqnarray}
If $a_2 \neq 0$ then 
$$\{X^5,X^5Y,X^6,X^6Y,X^7\} \subset \Box_{\prec_w}(F).$$
Hence, assume $a_2=0$. But then if $a_3 \neq a_4$ we get
\begin{equation}
\{X^3Y,X^4Y,X^5Y,X^6Y\} \subset \Box_{\prec_w}(F).\label{eqelefant}
\end{equation}
Multiplying the above polynomial by $Y$ and continuing the reduction
we obtain:
\begin{eqnarray}
&&(a_3+a_4)X^3Y^2+a_5XY^3+a_6X^2Y^2+a_7Y^3\nonumber \\
&&a_8XY^2+X^2Y+a_9Y^2+a_3XY\nonumber \\
&\overset{F(X,Y)}{\longrightarrow}&
                                    ((a_3+a_4)X^2+a_5Y+a_6X+a_8)\nonumber
  \\
&&(X^4a_3Y^2+a_4X^3+a_5XY+a_6X^2+a_7Y+a_8X+a_9) \nonumber \\
&&+a_7Y^3+a_8XY^2+X^2Y+a_9Y^2+a_3XY\nonumber
\end{eqnarray}
implying that 
$$\{ X^6,X^7\} \subset \Box_{\prec_w}(F).$$
Hence, for the case $a_1=1$, $a_2=0$, $a_3 \neq a_4$ in addition
to~(\ref{eqsnabeleins}) we found $6$ more elements in $\Box_{\prec_w}(F)$. Namely,
the above $2$ and the $4$ in~(\ref{eqelefant}).

In the following we assume $a_3=a_4$. Continuing the reduction we
obtain
\begin{eqnarray}
&\overset{F(X,Y)}{\longrightarrow}&a_5X^4+a_6X^2Y+(a_4a_5+a_7)Y^2+a_4a_5X^3+(a_5^2+a_8)XY\nonumber\\
&&+(a_5a_6+1)X^2+(a_5a_7+a_9)Y+(a_5a_8+a_4)X+a_5a_9.\label{eqhererdet} 
\end{eqnarray}
If $a_5 \neq 0$ then
$$\{X^4,X^4Y,X^5,X^5Y,X^6,X^6Y,X^7\} \subset \Box_{\prec_w}(F).$$
Hence, we next assume $a_5 = 0$. If then $a_6 \neq 0$ we obtain
$$\{X^2Y,X^3Y,X^4Y,X^5Y,X^6Y\} \subset \Box_{\prec_w}(F),$$
and we therefore now assume $a_6=0$. We next multiply the considered
polynomial by $X$ and continue the reduction
\begin{eqnarray}
&&a_7XY^2+a_8X^2Y+X^3+a_9XY+a_4X^2\nonumber \\
&\overset{F(X,Y)}{\longrightarrow}&a_7X^4+a_8X^2Y+a_4a_7Y^2+(a_4a_7+1)X^3
                                   \nonumber \\
&&+a_9XY+a_4X^2+a_7^2Y+a_7a_8X+a_7a_9.\nonumber
\end{eqnarray}
If $a_7 \neq 0$ then 
$$\{Y^2,X^4,X^4Y,X^5,X^5Y,X^6,X^6Y,X^7\} \subset \Box_{\prec_w}(F).$$
Here -- although it has no implication for what we want to prove --
we used (\ref{eqhererdet}) to demonstrate that $Y^2$ is also in the
set. Hence, assume now that $a_7=0$. Then if $a_8 \neq 0$ we obtain
$$\{X^2Y,X^3Y,X^4Y,X^5Y,X^6Y\} \subset \Box_{\prec_w}(F).$$
Finally, if $a_8=0$ the leading monomial becomes $X^3$ and we
therefore have
$$\{X^3,X^3Y,X^4,X^4Y,X^5,X^5Y,X^6,X^6Y,X^7\} \subset
\Box_{\prec_w}(F).$$
In conclusion we have established the existence of at least $6+\min \{
4,5,6,7,5,8,5\}=10$ elements in $\Box_{\prec_w}(F)$, and therefore $w_H(\vec{c})
\geq 10$. Moreover, by inspection of the results in the present
section we see that 
$w_H(\vec{c}) \geq 6+5=11$
holds when $a_1\in \{0,1\}$.

\subsection{Leading monomial equal to $X^3Y$}
Consider $\vec{c}={\mbox{ev}}(F+I_8)$ where 
\begin{eqnarray}
F(X,Y)&=&X^3Y+a_1XY^2+a_2X^4+a_3X^2Y+a_4Y^2\nonumber \\
&&+a_5X^3+a_6XY+a_7X^2+a_8Y +a_9X+a_{10}.
\nonumber 
\end{eqnarray}
For sure
\begin{eqnarray}
&&\{ X^3Y,X^4Y,X^5Y,X^6Y,X^3Y^2,X^4Y^2,X^5Y^2,X^6Y^2\} \subset \Box_{\prec_w}(F).\nonumber
\end{eqnarray}
We next consider an exhaustive series of conditions under which we establish more
monomials in $\Box_{\prec_w}(F)$. The strategy in this subsection is
different from other sections in that we here do not reduce modulo
$F(X,Y)$ but instead in addition to reducing modulo $Y^3+X^3Y+X$ also reduce modulo the
polynomials  $X^8+X, X^7Y+Y \in \langle
Y^3+X^3Y+X,Y^8+Y,X^8+X\rangle$.\\

We start by multiplying $F(X,Y)$ by $X^7$ to obtain
\begin{eqnarray}
&&X^{10}Y+a_1X^8Y^2+a_2X^{11}+a_3X^9Y+a_4X^7Y^2+a_5X^{10}+a_6X^8Y\nonumber
   \\
&&+a_7X^9+a_8X^7Y+a_9X^8+a_{10}X^7\nonumber \\
&\overset{X^8+X}{\longrightarrow}&X^3Y+a_1XY^2+a_2X^4+a_3X^2Y+a_4X^7Y^2\nonumber
  \\
&&+a_5X^3+a_6XY+a_7X^2+a_8X^7Y+a_9X+a_{10}X^7\nonumber \\
&\overset{X^7Y+Y}{\longrightarrow}&X^3Y+a_1XY^2+a_2X^4+a_3X^2Y+a_4Y^2\nonumber
  \\
&&+a_5X^3+a_6XY+a_7X^2+a_8Y+a_5X+a_{10}X^7.\nonumber 
\end{eqnarray} 
If $a_{10} \neq 0$ then 
$$\{ X^7\} \subset \Box_{\prec_w}(F).$$
Hence, assume $a_{10}=0$ and multiply the resulting expression by
$Y^2$ to obtain
\begin{eqnarray}
&&X^3Y^3+a_1XY^4+a_2X^4Y^2+a_3X^2Y^3+a_4Y^4\nonumber \\
&&+a_5X^3Y^2+a_6XY^3+a_7X^2Y^2+a_8Y^3+a_9XY^2\nonumber \\
&\overset{Y^3+X^3Y+X}{\longrightarrow}&X^6Y+X^4+a_1X^4Y^2+a_1X^2Y+a_2X^4Y^2+a_3X^5Y\nonumber
  \\
&&+a_3X^3+a_4X^3Y^2+a_4XY+a_5X^3Y^2+a_6X^4Y+a_6X^2\nonumber \\
&&+a_7X^2Y^2+a_8X^3Y+a_8X+a_9XY^2\nonumber 
\end{eqnarray}
which we multiply by $X^6$ to obtain
\begin{eqnarray}
&&X^{12}Y+X^{10}+a_1X^{10}Y^2+a_1X^8Y+a_2X^{10}Y^2+a_3X^{11}Y+a_3X^9
\nonumber
  \\
&&+a_4X^9Y^2+a_4X^7Y
+a_5X^9Y^2+a_6X^{10}Y+a_6X^8+a_7X^8Y^2\nonumber \\
&&
+a_8X^9Y+a_8X^7+a_9X^7Y^2
\nonumber
  \\
&\overset{X^7Y+Y}{\longrightarrow}& \cdots \nonumber
  \\&\overset{X^8+X}{\longrightarrow}&
                                       X^5Y+X^3+a_1X^3Y^2+a_1XY+a_2X^3Y^2+a_3X^4Y+a_3X^2\nonumber
                                       \\
&&+a_4X^2Y^2+a_4Y+a_5X^2Y^2+a_6X^3Y+a_6X+a_7XY^2+a_8X^2Y\nonumber \\
&&+a_8X^7+a_9Y^2.\nonumber 
\end{eqnarray}
If $a_8\neq 0$ then 
$$\{X^7\} \subset \Box_{\prec_w}(F).$$
Hence, assume $a_8=0$. We next multiply $F(X,Y)$ by $Y$ and obtain
\begin{eqnarray}
&&X^3Y^2+a_1XY^3+a_2X^4Y+a_3X^2Y^2+a_4Y^3+a_5X^3Y\nonumber \\
&&+a_6XY^2+a_7X^2Y+a_9XY\nonumber
  \\
&\overset{Y^3+X^3Y+X}{\longrightarrow}&X^3Y^2+a_1X^4Y+a_1X^2+a_2X^4Y+a_3X^2Y^2\nonumber
  \\
&&+a_4X^3Y+a_4X+a_5X^3Y+a_6XY^2+a_7X^2Y+a_9XY\nonumber
\end{eqnarray}
which we multiply by $X^6$ to obtain
\begin{eqnarray}
&&X^9Y^2+a_1X^{10}Y+a_1X^8+a_2X^{10}Y+a_3X^8Y^2\nonumber \\
&&+a_4X^9Y+a_4X^7+a_5X^9Y+a_6X^7Y^2+a_7X^8Y+a_9X^7Y\nonumber \\
&\overset{X^7Y+Y}{\longrightarrow}&X^2Y^2+a_1X^3Y+a_1X+a_2X^3Y+a_3XY^2\nonumber
  \\
&&+a_4X^2Y+a_4X^7+a_5X^2Y+a_6Y^2+a_7XY+a_9Y.\nonumber 
\end{eqnarray}
If $a_4 \neq 0$ then
$$\{X^7\} \subset \Box_{\prec_w}(F).$$
Hence, assume $a_4=0$. We next multiply $F(X,Y)$ by $X^6$ to obtain
\begin{eqnarray}
&&X^9Y+a_1X^7Y^2+a_2X^{10}+a_3X^8Y+a_5X^9\nonumber \\
&&+a_6X^7Y+a_7X^8+a_9X^7
   \nonumber \\
&\overset{X^7Y+Y}{\longrightarrow}&\cdots \nonumber \\
&\overset{X^8+X}{\longrightarrow}&X^2Y+a_1Y^2+a_2X^3+a_3XY+a_5X^2+a_6Y+a_7X+a_9X^7. \nonumber 
\end{eqnarray}
If $a_9 \neq 0$ then 
$$\{X^7\} \subset \Box_{\prec_w}(F).$$
Hence, assume $a_9=0$ and multiply by $Y^2$
\begin{eqnarray}
&&X^2Y^3+a_1Y^4+a_2X^3Y^2+a_3XY^3+a_5X^2Y^2+a_6Y^3+a_7XY^2 \nonumber
  \\
&\overset{Y^3+X^3Y+X}{\longrightarrow}&X^5Y+X^3+a_1X^3Y^2+a_1XY+a_2X^3Y^2+a_3X^4Y\nonumber
  \\
&&+a_3X^2+a_5X^2Y^2+a_6X^3Y+a_6X+a_7XY^2 \nonumber
\end{eqnarray}
which we then multiply by $X^6$ to obtain
\begin{eqnarray}
&&X^{11}Y+X^9+a_1X^9Y^2+a_1X^7Y+a_2X^9Y^2+a_3X^{10}Y+a_3X^8\nonumber
  \\
&&+a_5X^8Y^2+a_6X^9Y+a_6X^7+a_7X^7Y^2\nonumber
  \\
&\overset{X^7Y+Y}{\longrightarrow}&\cdots \nonumber \\
&\overset{X^8+X}{\longrightarrow}&X^4Y+X^2+a_1X^2Y^2+a_1Y+a_2X^2Y^2+a_3X^3Y+a_3X\nonumber
                                   \\
&&+a_5XY^2+a_6X^2Y+a_6X^7+a_7Y^2.\nonumber 
\end{eqnarray}
If $a_6 \neq 0$ then 
$$\{X^7\} \subset \Box_{\prec_w}(F).$$
Hence, assume $a_6=0$. We next multiply $F(X,Y)$ by $Y$ and continue
the reductions:
\begin{eqnarray}
&&X^3Y^2+a_1XY^3+a_2X^4Y+a_3X^2Y^2+a_5X^3Y+a_7X^2Y\nonumber \\
&\overset{Y^3+X^3Y+X}{\longrightarrow}&X^3Y^2+a_1X^4Y+a_1X^2+a_2X^4Y+a_3X^2Y^2+a_5X^3Y+a_7X^2Y                                        \nonumber
\end{eqnarray}
which we multiply by $X^5$
\begin{eqnarray}
&&X^8Y^2+a_1X^9Y+a_1X^7+a_2X^9Y+a_3X^7Y^2+a_5X^8Y+a_7X^7Y\nonumber \\
&\overset{X^7Y+Y}{\longrightarrow}&XY^2+a_1X^2Y+a_1X^7+a_2X^2Y+a_3Y^2+a_5XY+a_7Y.\nonumber 
\end{eqnarray}
If $a_1 \neq 0$ then 
$$\{X^7\} \subset \Box_{\prec_w}(F).$$
Hence, assume $a_1=0$. We next multiply $F(X,Y)$ by $X^5$
\begin{eqnarray}
&&X^8Y+a_2X^9+a_3X^7Y+a_5X^8+a_7X^7 \nonumber \\
&\overset{X^7Y+Y}{\longrightarrow}& \cdots \nonumber \\
&\overset{X^8+X}{\longrightarrow}& XY+a_2X^2+a_3Y+a_5X+a_7X^7. \nonumber 
\end{eqnarray}
If $a_7 \neq 0$ then 
$$\{ X^7\} \subset \Box_{\prec_w}(F).$$
Hence, assume $a_7=0$. Next we multiply $F(X,Y)$ by $Y^2$ and obtain
\begin{eqnarray}
&&X^3Y^3+a_2X^4Y^2+a_3X^2Y^3+a_5X^3Y^2 \nonumber \\
&\overset{Y^3+X^3Y+X}{\longrightarrow}&X^6Y+X^4+a_2X^4Y^2+a_3X^5Y+a_3X^3+a_5X^3Y^2\nonumber 
\end{eqnarray}
which we multiply by $X^4$
\begin{eqnarray}
&&X^{10}Y+X^8+a_2X^8Y^2+a_3X^9Y+a_3X^7+a_5X^7Y^2\nonumber \\
&\overset{X^7Y+Y}{\longrightarrow}&\cdots  \nonumber \\
&\overset{X^8+X}{\longrightarrow}&X^3Y+X+a_2XY^2+a_3X^2Y+a_3X^7+a_5Y^2. \nonumber
\end{eqnarray}
If $a_3 \neq 0$ then 
$$\{X^7\} \subset \Box_{\prec_w}(F).$$
Hence, assume $a_3=0$. We now multiply $F(X,Y)$ by $X^4$
\begin{eqnarray}
&&X^7Y+a_2X^8+a_5X^7\nonumber \\
&\overset{X^7Y+Y}{\longrightarrow}&\cdots \nonumber \\
&\overset{X^8+x}{\longrightarrow}&Y+a_2X+a_5X^7. \nonumber 
\end{eqnarray}
If $a_5\neq 0$ then
$$\{X^7\} \subset \Box_{\prec_w}.$$
Hence, assume finally that $a_5 \neq 0$ and multiply $F(X,Y)$ by $Y^2$
to obtain
\begin{eqnarray}
&&X^3Y^3+a_2X^4Y^2\nonumber \\
&\overset{Y^3+X^3Y+X}{\longrightarrow}&X^6Y+X^4+a_2X^4Y^2.\nonumber
\end{eqnarray}
This expression is then multiplied by $X^3$
\begin{eqnarray}
&&X^9Y+X^7+a_2X^7Y^2\nonumber \\
&\overset{X^7Y+Y}{\longrightarrow}&X^2Y+X^7+a_2Y^2 \nonumber
\end{eqnarray}
and 
$$\{X^7\} \subset \Box_{\prec_w}\subset \Box_{\prec_w}(F).$$ 
In conclusion $w_H(\vec{c})\geq 8+1=9$.

\subsection{Leading monomial equal to $X^2Y^2$}
Consider $\vec{c}={\mbox{ev}}(F+I_8)$ where 
\begin{eqnarray}
F(X,Y)&=&X^2Y^2+a_1X^5+a_2X^3Y+a_3XY^2+a_4X^4+a_5X^2Y\nonumber \\
&&+a_6Y^2+a_7X^3+a_8XY+a_9X^2+a_{10}Y +a_{11}X+a_{12}.
\nonumber 
\end{eqnarray}
For sure
\begin{eqnarray}
&&\{ X^2Y^2,X^3Y^2,X^4Y^2,X^5Y^2,X^6Y^2\} \subset \Box_{\prec_w}(F).\nonumber
\end{eqnarray}
We next consider an exhaustive series of conditions under which we establish more
monomials in $\Box_{\prec_w}(F)$. We have
\begin{eqnarray}
&&YF(X,Y) \nonumber \\
&\overset{Y^3+X^3Y+X}{\longrightarrow}&(1+a_1)X^5Y+a_2X^3Y^2+a_3XY^3+a_4X^4Y+a_5X^2Y^2\nonumber
  \\
&&+a_6Y^3+a_7X^3Y+a_8XY^2+a_9X^2Y+a_{10}Y^2+X^3\nonumber \\
&&+a_{11}XY+a_{12}Y.\nonumber
\end{eqnarray}
If $a_1\neq 1$ then 
$$\{ X^5Y,X^6Y\} \subset \Box_{\prec_w}(F).$$
Hence, assume $a_1=1$ and continue the reduction.
\begin{eqnarray}
&\overset{Y^3+X^3Y+X}{\longrightarrow}&a_2X^3Y^2+(a_3+a_4)X^4Y+a_5X^2Y^2+(a_6+a_7)X^3Y+a_8XY^2\nonumber
  \\
&&+a_9X^2Y+a_{10}Y^2+X^3+a_{11}XY+a_3X^2+a_{12}Y+a_6X\nonumber \\
&\overset{F(X,Y)}{\longrightarrow}&a_2X^6+(a_3+a_4+a_2^2)X^4Y+(a_2a_3+a_5)X^2Y^2+a_2a_4X^5 \nonumber \\
&&+(a_2a_5+a_6+a_7)X^3Y+(a_2a_6+a_8)XY^2+a_2a_7X^4\nonumber \\
&&+(a_2a_8+a_9)X^2Y+a_{10}Y^2+(a_2a_9+1)X^3+(a_2a_{10}+a_{11})XY\nonumber
  \\
&&+(a_2a_{11}+a_3)X^2+a_{12}Y+(a_2a_{12}+a_6)X.\nonumber
\end{eqnarray}
If $a_2  \neq 0$ then
$$\{X^6,X^6Y,X^7\} \subset \Box_{\prec_w}(F).$$
Hence, assume $a_2=0$. If $a_3\neq a_4$ then we have 
$$\{X^4Y,X^5Y,X^6Y\} \subset \Box_{\prec_w}(F).$$
Assuming $a_3=a_4$ we continue the reduction as follows
\begin{eqnarray}
&\overset{F(X,Y)}{\longrightarrow}&a_5X^5+(a_6+a_7)X^3Y+(a_4a_5+a_8)XY^2+a_4a_5X^4+(a_5^2+a_9)X^2Y\nonumber  \\
&& +(a_5a_6+a_{10})Y^2+(a_5a_7+1)X^3+(a_5a_8+a_{11})XY+(a_5a_9+a_4)X^2\nonumber
   \\
&&+(a_5a_{10}+a_{12})Y
   +(a_5a_{11}+a_6)X+a_5a_{12}.\nonumber 
\end{eqnarray} 
If $a_5 \neq 0$ then 
$$\{X^5,X^5Y,X^6, X^6Y,X^7\} \subset \Box_{\prec_w}(F).$$
Hence, assume $a_5=0$. But then if $a_6 \neq a_7$
$$\{X^3Y,X^4Y,X^5Y,X^6Y\} \subset \Box_{\prec_w}(F),$$
and we therefore next assume $a_6=a_7$. We now multiply the above
polynomial by $X$ and continue the reduction
\begin{eqnarray}
&&X\big( a_8XY^2+a_9X^2Y+a_{10}Y^2+X^3+a_{11}XY+a_4X^2+a_{12}Y+a_6X\big)\nonumber \\
&\overset{F(X,Y)}{\longrightarrow}&a_8X^5+a_9X^3Y+(a_4a_8+a_{10})XY^2+(a_4a_8+1)X^4+a_{11}X^2Y\nonumber
  \\
&&+a_7a_8Y^2+(a_4+a_7a_8)X^3+(a_8^2+a_{12})XY+(a_6+a_8a_9)X^2\nonumber
   \\
&&+a_8a_{10}Y+a_8a_{11}X+a_8a_{12}.\nonumber
\end{eqnarray}
If $a_8 \neq 8$ then
$$\{X^5,X^5Y,X^6,X^6Y,X^7\} \subset \Box_{\prec_w}(F),$$
and if $a_8=0$ but $a_9 \neq 0$ then 
$$\{X^3Y,X^4Y,X^5Y,X^6Y\}\subset \Box_{\prec_w}(F).$$
Hence, assume $a_8=a_9=0$ and multiply the resulting polynomial by $X$
after which we continue the reduction.
\begin{eqnarray}
&&X \big(a_{10}XY^2+X^4+a_{11}X^2Y+a_4X^3+a_{12}XY+a_6X^2\big)
   \nonumber \\
&\overset{F(X,Y)}{\longrightarrow}&(a_{10}+1)X^5+a_{11}X^3Y+a_4a_{10}XY^2+(a_4+a_4a_{10})X^4+a_{12}X^2Y\nonumber
  \\
&&+a_7a_{10}Y^2+(a_7a_{10}+a_7)X^2+a_{10}^2Y+a_{10}a_{11}X+a_{10}a_{12}.\nonumber
\end{eqnarray}
If $a_{10} \neq 1$ then 
$$\{X^5,X^5Y,X^6,X^6Y,X^7\} \subset \Box_{\prec_w}(F).$$
Hence, assume $a_{10}=1$. If $a_{11}\neq 0$ then 
$$\{X^3Y,X^4Y,X^5Y,X^6Y\} \subset \Box_{\prec_w}(F).$$
Hence assume $a_{11}=0$ and multiply the resulting polynomial by $X$
and continue the reduction
\begin{eqnarray}
&&X\big( a_4XY^2+a_{12}X^2Y+a_7Y^2+Y+a_{12} \big)\nonumber \\
&\overset{F(X,Y)}{\longrightarrow}&a_4X^5+a_{12}X^3Y+(a_4^2+a_7)XY^2+a_4^2X^4+a_4a_7Y^2\nonumber
  \\
&&+a_4a_7X^3+XY+a_4Y+a_{12}X+a_4a_{12}.\nonumber
\end{eqnarray}
If $a_4 \neq 0$ then
$$\{X^5,X^5Y,X^6,X^6Y,X^7\} \subset \Box_{\prec_w}(F).$$
Hence, assume $a_4=0$ Then if $a_{12} \neq 0$ we have
$$\{X^3Y,X^4Y,X^5Y,X^6Y\} \subset \Box_{\prec_w}(F).$$
Hence, we assume $a_{12}=0$. We again multiply by $X$ and continue the
reduction
\begin{eqnarray}
&X \big( a_7XY^2+XY\big)\overset{F(X,Y)}{\longrightarrow} a_7X^5+a_7^2Y^2+a_7^2X^3+a_7Y+X^2Y.\nonumber
\end{eqnarray}
If $a_7 \neq 0$ then 
$$\{X^5,X^5Y,X^6,X^6Y,X^7\} \subset \Box_{\prec_w}(F).$$
Finally, assume $a_7=0$. But then we are left with $X^2Y$ and
therefore
$$\{X^2Y,X^3Y,X^4Y,X^5Y,X^6Y\} \subset \Box_{\prec_w}(F).$$
In conclusion we can always establish at least $5+\min \{ 2, 3,3,5,4,5,4,5,4,5,4\} = 7$ monomials in
$\Box_{\prec_w}(F)$, and we conclude that 
$w_H(\vec{c})\geq 7$. 
Moreover, our analysis reveals that if $a_1=1$ then $w_H(\vec{c}) \geq 5+3=8$.

\subsection{Leading monomial equal to $X^3Y^2$}
Consider $\vec{c}={\mbox{ev}}(F+I_8)$ where 
\begin{eqnarray}
F(X,Y)&=&X^3Y^2+a_1X^6+a_2X^4Y+a_3X^2Y^2+a_4X^5\nonumber \\
&&+a_5X^3Y+a_6XY^2+a_7X^4+a_8X^2Y+a_9Y^2\nonumber \\
&&+a_{10}X^3+a_{11}XY+a_{12}X^2+a_{13}Y+a_{14}X+a_{15}.\nonumber
\end{eqnarray}
For sure
$$
\{ X^3Y^2,X^4Y^2,X^5Y^2,X^6Y^2\} \subset \Box_{\prec_w}(F). 
$$
We next consider an exhaustive series of conditions under which we establish more
monomials in $\Box_{\prec_w}(F)$. We have
\begin{eqnarray}
&&YF(X,Y) \nonumber \\
&\overset{Y^3+X^3Y+X}{\longrightarrow}&(1+a_1)X^6Y+a_2X^4Y^2+a_3X^2Y^3+a_4X^5Y+a_5X^3Y^2\nonumber
  \\
&&+a_6XY^3+a_7X^4Y+a_8X^2Y^2+a_9Y^3+a_{10}X^3Y+a_{11}XY^2\nonumber \\
&&+X^4+a_{12}X^2Y +a_{13}Y^2+a_{14}XY+a_{15}Y. \nonumber
\end{eqnarray}
If $a_1\neq 1$ then $$\{X^6Y\} \subset \Box_{\prec_w}(F).$$
Hence, assume $a_1=1$ and continue the reduction.
\begin{eqnarray}
&\overset{Y^3+X^3Y+X}{\longrightarrow}&a_2X^4Y^2+(a_3+a_4)X^5Y+a_5X^3Y^2+(a_6+a_7)X^4Y\nonumber
  \\
&&+a_8X^2Y^2+(a_9+a_{10})X^3Y+a_{11}XY^2+X^4+a_{12}X^2Y\nonumber \\
&&+a_{13}Y^2+a_3X^3+a_{14}XY+a_6X^2+a_{15}Y+a_9X\nonumber \\
&\overset{F(X,Y)}{\longrightarrow}&a_2X^7+(a_2^2+a_3+a_4)X^5Y+(a_2a_3+a_5)X^3Y^2+a_2a_4X^6\nonumber
  \\
&&+(a_2a_5+a_6+a_7)X^4Y+(a_2a_6+a_8)X^2Y^2+a_2a_7X^5 \nonumber \\
&&+(a_2a_8+a_9+a_{10})X^3Y+(a_2a_9+a_{11})XY^2 +(a_2a_{10}+1)X^4 \nonumber \\
&&+(a_2a_{11}+a_{12})X^2Y+a_{13}Y^2+(a_2a_{12}+a_3)X^3+(a_2a_{13}+a_{14})XY\nonumber
  \\
&&+(a_2a_{14}+a_6)X^2+a_{15}Y+(a_2a_{15}+a_9)X.\nonumber
\end{eqnarray}
If $a_2 \neq 0$ then 
$$\{X^7\} \subset \Box_{\prec_w}(F).$$
Hence, assume $a_2=0$. If $a_3 \neq a_4$ then
$$\{X^5Y,X^6Y\} \subset \Box_{\prec_w}(F).$$
Hence, assume $a_3=a_4$ and continue the reduction. 
\begin{eqnarray}
&\overset{F(X,Y)}{\longrightarrow}&a_5X^6+(a_6+a_7)X^4Y+(a_4a_5+a_8)X^2Y^2+a_4a_5X^5\nonumber
  \\
&&+(a_5^2+a_9+a_{10})X^3Y+(a_5a_6+a_{11})XY^2+(a_5a_7+1)X^4\nonumber
  \\
&&+(a_5a_8+a_{12})X^2Y+(a_5a_9+a_4)Y^2+(a_5a_{10}+a_4)X^3   \nonumber \\
&&+(a_5a_{11}+a_{14})XY+(a_5a_{12}+a_6)X^2
+(a_5a_{13}+a_{15})Y\nonumber \\
&&+(a_5a_{14}+a_9)X+a_5a_{15}. \nonumber
\end{eqnarray}
If $a_5 \neq 0$ then
$$\{X^6,X^6Y,X^7\} \subset \Box_{\prec_w}(F).$$
Hence, assume $a_5=0$. But then if $a_6\neq a_7$ 
$$\{X^4Y,X^5Y,X^6Y\} \subset \Box_{\prec_w}(F).$$
Hence, assume $a_6=a_7$. But then if $a_8 \neq 0$ we obtain
$$\{X^2Y^2 \} \subset \Box_{\prec_w}(F).$$
Actually, this result could be improved to
$$\{X^2Y^2,X^6,X^6Y,X^7\} \subset \Box_{\prec_w}(F)$$
if we multiply the above polynomial by $X$ and reduce it modulo
$F(X,Y)$. The details  are left for the reader. Next assume $a_8=0$. But then if $a_9 \neq a_{10}$ we get
$$\{X^3Y,X^4Y,X^5Y,X^6Y\} \subset \Box_{\prec_w}(F).$$
Hence, assume $a_9=a_{10}$. If $a_{11} \neq 0$ then
$$\{XY^2,X^2Y^2\}\subset \Box_{\prec_w}(F).$$
Finally, assume $a_{11}=0$. But then $X^4$ is the leading monomial and
we obtain
$$\{X^4,X^4Y,X^5,X^5Y,X^6,X^6Y,X^7\} \subset \Box_{\prec_w}(F).$$
In conclusion we can always establish at least $4+\min \{ 1, 1,2,3,3,4,4,2,7\} = 5$ monomials in
$\Box_{\prec_w}(F)$, and we conclude that 
$w_H(\vec{c})\geq 5$. Moreover, if $a_1 = 1$ and $a_2 = 0$ then
$w_H(\vec{c}) \geq 4+2=6$.

\subsection{Leading monomial equal to $X^7$}\label{subs}
Consider $\vec{c}={\mbox{ev}}(F+I_8)$ where 
\begin{eqnarray}
F(X,Y)&=&X^7+a_1X^5Y+a_2X^3Y^2+a_3X^6+a_4X^4Y+a_5X^2Y^2\nonumber \\
&&+a_6X^5+a_7X^3Y+a_8XY^2+a_9X^4+a_{10}X^2Y+a_{11}Y^2\nonumber \\
&&+a_{12}X^3+a_{13}XY+a_{14}X^2+a_{15}Y+a_{16}X+a_{17}. \nonumber
\end{eqnarray}
Observe that among the $22$ affine roots over ${\mathbb{F}}_8$ of $Y^3+X^3Y+X$ the
only point having the first coordinate equal to $0$ is $(0,0)$. Hence,
${\mbox{ev}}(X^7+1)$ is of Hamming weight $1$ meaning that  $w_H(\vec{c})=1$ when $a_1=\cdots = a_{16}=0$
and $a_{17}=1$. In the following we show that for all other choices of
$a_i$ the Hamming weight becomes  at least $3$. We first observe, that 
$$\{X^7\} \subset \Box_{\prec_w}(F).$$
Now consider
\begin{eqnarray}
&&YF(X,Y) \nonumber \\
&\overset{X^7Y+Y}{\longrightarrow}&a_1X^5Y^2+a_2X^3Y^3+a_3X^6Y+a_4X^4Y^2+a_5X^2Y^3
                                    \nonumber \\
&&+a_6X^5Y+a_7X^3Y^2+a_8XY^3+a_9X^4Y+a_{10}X^2Y^2\nonumber \\
&&
+a_{11}Y^3+a_{12}X^3Y+a_{13}XY^2+a_{14}X^2Y\nonumber \\
&&+a_{15}Y^2+a_{16}XY+(a_{17}+1)Y\nonumber \\
&\overset{Y^3+X^3Y+X}{\longrightarrow}& a_1
                                        X^5Y^2+(a_2+a_3)X^6Y+a_4X^4Y^2+(a_5+a_6)X^5Y
                                        \nonumber \\
&&+a_7X^3Y^2+(a_8+a_9)X^4Y+a_{10}X^2Y^2+(a_{11}+a_{12})X^3Y
   \nonumber \\
&&+a_{13}XY^2+a_2X^4+a_{14}X^2Y+a_{15}Y^2+a_5X^3+a_{16}XY\nonumber \\
&&+a_8X^2+(a_{17}+1)Y+a_{11}X.\nonumber
\end{eqnarray}
If the above polynomial is non-zero then going through all possible
leading monomials we see that we can always establish at least two
more monomials in $\Box_{\prec_w}(F)$ in addition to $X^7$. For instance if $a_1 \neq 0$ then
we can add $\{X^5Y^2,X^6Y^2\}$. If $a_1=0$ and $a_2 \neq a_3$ then we
can add $\{X^6Y,X^6Y^2\}$ and so on. By inspection the above
polynomial equals the zero polynomial if and only if $F(X,Y)=X^7+1$
and we are through. 
\subsection{The remaining cases}
For the remaining choices of leading monomial it seems impossible to
obtain better information on 
$\Box_{\prec_w}(F )$ than what is derived by noting that all
monomials divisible by ${\mbox{lm}}(F)$ must be a leading monomial in $\langle F
\rangle +I_8$. In particular when the leading monomial is $X^i$, $i=0, \ldots ,
7$ the information we obtain in this way can be shown to be the true Hamming weight
of existing corresponding codewords. In conclusion we established the
information in Figure~\ref{figto}.

\begin{figure}
\begin{center}
$$
\begin{array}{cccccccc}
{\mbox{ }}13{\mbox{ }}&{\mbox{ }}10{\mbox{ }}&{\mbox{ }}7{\mbox{ }}&{\mbox{ }}5{\mbox{ }}&{\mbox{ }}3{\mbox{ }}&{\mbox{ }}2{\mbox{ }}&{\mbox{ }}1{\mbox{ }}
\\
18&15&12&9&6&4&2\\
22&19&16&13&10&7&4&{\mbox{ }}1
\end{array}
$$
\end{center}
\caption{Lower bounds on $\# \Box_{\prec_w}(F)$ where ${\mbox{lm}}(F)$
  are as in Figure~\ref{figfirstone}}
\label{figto}
\end{figure}
 
\section{Code parameters}\label{seccodes}

As code construction we use
\begin{eqnarray}
{\mbox{Span}}_{{\mathbb{F}}_8} \{ {\mbox{ev}}(M+I_8)
                \mid M \in \Delta_{\prec_w}(I_8), \delta(M) \geq s\},
                \nonumber 
\end{eqnarray} 
where $\delta(M)$ are the estimates of $\# \Box_{\prec_w}(F)$ as
depicted in Figure~\ref{figto}. In this way we obtain the best
possible codes, according to our estimates. The resulting parameters
are shown in 
Table~\ref{tab3}.
\begin{table}
\begin{center}
\begin{tabular}{ccc}
$[22,1,22]_8$&$[22,2,19]_8$&$[22,3,18]_8$\\
$[22,4,16]_8$&$[22,5,15]_8$&$[22,7,13]_8$\\
$[22,8,12]_8$&$[22,10,10]_8$&$[22,11,9]_8$\\
$[22,13,7]_8$&$[22,14,6]_8$&$[22,15,5]_8$\\
$[22,17,4]_8$&$[22,18,3]_8$&$[22,20,2]_8$
\end{tabular}
\end{center}
\caption{Parameters $[n,k,d]_8$ of codes from the Klein
  quartic. Here, $n$ and $k$ are sharp values, whereas $d$ represents a
  lower bound estimate.}
\label{tab3}
\end{table}
In almost all cases, given a dimension in the table, then the 
corresponding estimate on the minimum distance equals the best 
value known to exist according to~\cite{grassl}. The only
exceptions are the dimensions $4, 14, 15$ and $18$ where the best
minimum distances known to exist are one more than we obtain. We finally remark that
if we evaluate in all polynomials except those who have $X^6Y^2$ in their
support then by Subsection~\ref{subs} we get a code of dimension $21$ with exactly $7$
codewords of Hamming weight $1$. Hence, this code is almost as good as
the $[22,21,2]]_8$ code, known
to exist by~\cite{grassl}.
\section{Concluding remarks}
In~\cite[Ex.\ 3.2]{KFR} the authors estimated the minimum
distances of the duals of the codes studied in the present paper  using the Feng-Rao bound for dual codes. We
believe that is should be possible to improve (possibly even drastic) upon their estimates of the
minimum distance in the same way as we in this paper improved upon the
Feng-Rao bound for primary codes. We leave this question for future
research. The method of the present paper also applies to
estimate higher weights (possible relative). We leave it for future
research to establish examples where this gives improved information
compared to what can be derived from the Feng-Rao bound. In the light
of Remark~\ref{remrom} and the information established in
Section~\ref{secnew}, evidently our new method sometimes 
significantly improves upon the previous known methods. We stress that our
method is very general in that it can be applied to any primary affine
variety code. In particular it works for any monomial ordering and
consequently also without any of the order domain conditions
(Remark~\ref{remrom}). Finding more families of good affine variety
codes using our method is subject to future work. 
\section*{Acknowledgments}
The authors gratefully acknowledge the support from The Danish Council for Independent Research (Grant
No.\ DFF--4002-00367). They are also grateful to Department of
Mathematical Sciences, Aalborg University for supporting a one-month
visiting professor position for the second listed author. The research of Ferruh \"{O}zbudak has been funded by METU Coordinatorship
of Scientific Research Projects via grant for projects BAP-01-01-2016-008 and
BAP-07-05-2017-007.

\def\cprime{$'$}

\end{document}